\documentclass[prl,twocolumn,floatix,amssymb,showpacs,amsmath,superscriptaddress]{revtex4-1}
\usepackage{graphicx}
\usepackage{epsfig,color}
\usepackage{amsmath,amssymb,amsthm}
\usepackage[utf8]{inputenc}

\newtheorem*{theo*}{Theorem}

\begin{document}

\title{Quantum Estimation Methods for Quantum Illumination}
\author{M. Sanz}
\email{mikel.sanz@ehu.eus}
\affiliation{Department of Physical Chemistry, University of the Basque Country UPV/EHU, Apartado 644, 48080 Bilbao, Spain}
\author{U. Las Heras} 
\affiliation{Department of Physical Chemistry, University of the Basque Country UPV/EHU, Apartado 644, 48080 Bilbao, Spain}
\author{J. J. Garc\'ia-Ripoll} 
\affiliation{Instituto de F\'isica Fundamental IFF-CSIC, Calle Serrano 113b, 28006 Madrid, Spain}
\author{E. Solano}
\affiliation{Department of Physical Chemistry, University of the Basque Country UPV/EHU, Apartado 644, 48080 Bilbao, Spain}
\affiliation{IKERBASQUE, Basque Foundation for Science, Maria Diaz de Haro 3, 48011 Bilbao, Spain}
\author{R. Di Candia}
\email{rob.dicandia@gmail.com}
\affiliation{Department of Physical Chemistry, University of the Basque Country UPV/EHU, Apartado 644, 48080 Bilbao, Spain}
\affiliation{Dahlem Center for Complex Quantum Systems, Freie Universit\"at Berlin, 14195 Berlin, Germany}

\begin{abstract}
Quantum illumination consists in shining quantum light on a target region immersed in a bright thermal bath, with the aim of detecting the presence of a possible low-reflective object. If the signal is entangled with the receiver, then a suitable choice of the measurement offers a gain with respect to the optimal classical protocol employing coherent states. Here, we tackle this detection problem by using quantum estimation techniques to measure the reflectivity parameter of the object, showing an enhancement in the signal-to-noise ratio up to $3$~dB with respect to the classical case when implementing only local measurements. Our approach employs the quantum Fisher information to provide an upper bound for the error probability, supplies the concrete estimator saturating the bound, and extends the quantum illumination protocol to non-Gaussian states. As an example, we show how Schr\"odinger's cat states may be used for quantum illumination.
\end{abstract}

\maketitle

{\it Introduction.---}
Entanglement is a necessary requirement for a number of quantum protocols, including quantum teleportation~\cite{Bennett93,DiCandia15}, superdense coding~\cite{Bennett92}, and quantum computation~\cite{Feynman82,DiVincenzo95}, among others. In his pioneering work~\cite{Lloyd08}, S. Lloyd showed how suitable entangled states can be used to detect the presence of a low-reflectivity object embedded in a bright environment, more efficiently than by using classical resources. This protocol, called quantum illumination (QI), consists of irradiating the target region by using a signal entangled with an ancilla, and optimally measuring the reflected signal together with the ancilla. Surprisingly, even if the final state is not entangled~\cite{Sacchi05}, the initial nonclassical correlations have a positive role in hypothesis testing performances. Lloyd's results, initially limited on a specific background noise, were extended considering a more general noise model~\cite{Shapiro091}. An alternative attempt was given by S.-H. Tan et al.~\cite{Tan08}, in which the authors prove the advantage of an entangled Gaussian state in the QI performance. In this case, a part of a two-mode squeezed state is sent to the target region, while the other copy remains in the lab. Then, an optimal {\it joint} measurement on the copies of the received signal and the ancillary modes generates a gain of at least 6~dB in the error probability decaying rate. Due to technical difficulties, they found only a lower bound for the decaying rate of the optimal error probability, namely Bhattacharyya bound~\cite{Audenaert07, Calsamiglia08}, by using tools specifically developed for Gaussian states~\cite{Pirandola08}. Although this proves the existence of a QI protocol showing certain gain with respect the classical case, the estimator achieving this is highly non-trivial, as it requires the implementation of a quantum Schur transform~\cite{Calsamiglia08}. In this sense, a local protocol is a desiderata, since it is simpler to be experimentally implemented.  Finally, a protocol consisting of {\it separate} measurements of the single copies of the reflected signal and the ancilla was found, showing a more modest 3~dB gain in the low photon regime~\cite{Guha09}. These results paved the way for relevant experimental applications within the purview of quantum radar~\cite{Lopaeva13, Barza15}, quantum communication~\cite{Shapiro09, Zhang13}, and quantum phase estimation~\cite{Zhang15,LasHeras16}, in which the unavoidable noise plays a crucial role.

In this Letter, we show that, fixed the number of photons, an entangled transmitter can improve the optimal estimation of the reflectivity parameter up to 3~dB with respect to a coherent state transmitter in the low-reflectivity limit. The optimal gain is achieved in the low-photon regime, and decays at least as the inverse of signal photon number. This is proven by bounding the quantum Fisher information (QFI) for a family of states labeled by the reflectivity parameter. We relate these results to the QI protocol, discussing a strategy based on the quantum estimation of the reflectivity parameter. We show that the QFI provides  a computable non-trivial upper bound on the optimal error probability, extending the QI protocol to non-Gaussian states. Our results are not limited by the usual QI constraints, since they can be applied to any bath's and signal's number of photons. Furthermore, this approach explicitly provides the concrete estimator attaining the proposed bound. The paper is structured in the following way. Firstly, we introduce the quantum estimation problem and compute the QFI. Then, we discuss our strategy for QI, providing the error probability bounds based on the QFI. Finally, we discuss two examples with Gaussian states and Schr\"odinger's cat states, showing that these states are also useful in QI.

{\it Quantum Estimation.---}
Let us consider a general bipartite pure state representation of the signal-idler system written in the Schmidt decomposition form
\begin{equation}
|\psi\rangle_{SI} = \sum_{\alpha}\sqrt{p_{\alpha}}|w_\alpha\rangle_S | v_\alpha\rangle_I,
\end{equation} 
where $\langle w_\alpha | w_{\alpha'}\rangle=\langle v_\alpha | v_{\alpha'}\rangle=\delta_{\alpha,\alpha'}$. In the QI protocol, the signal modes of the $M$ copies of $|\psi\rangle_{SI}$ are sent to the target region embedded in a bright thermal noise, in which there could possibly be an object. We then receive back $M$ copies of the state $\rho_\eta=\text{Tr}_S\,\left(U_\eta|\psi\rangle_{SI}\langle\psi|\otimes \rho_{B}U_{\eta}^\dag\right)$. Here, $U_\eta=\exp\left[\sin^{-1}(\eta)(s^\dag b-s b^\dag)\right]\simeq \exp\left[\eta(s^\dag b-s b^\dag)\right]$ is the signal-object interaction, modeled as a beamsplitter with amplitude reflectivity $\eta\ll1$, and $\rho_{B}=\sum_{n}\frac{N_B^n}{(1+N_B)^{1+n}}|n\rangle\langle n|$ is a thermal state with mean photon number $N_B$, as depicted in Fig.~\ref{Protocol}. In this framework,  the case $\eta=0$ corresponds to the absence of the object in the target region. In the following, we emphasize the case $N_B\gg1$, corresponding to the typical regime where QI shows a gain with respect the classical case. However, our treatment is completely general and it holds for any value of $N_B$. In order to optimally estimate $\eta$, the quantum Fisher information (QFI)~\cite{Paris09,Toth14} is a paradigmatic tool. This is due to the Cram\'er-Rao bound~\cite{Cramer46}, asserting the limits on the achievable precision of an unbiased estimator $\hat \eta$:
\begin{equation}
\Delta \hat \eta^2 \geq \frac{1}{M H},\label{CR}
\end{equation}
where $H= 2\sum_{mn}\frac{\left|\langle \phi_m|(\partial_\eta\rho_{\eta})|_{\eta=0}|\phi_n\rangle \right|^2}{\lambda_m+\lambda_n}$ is the QFI for the family of states $\rho_\eta$, $\lambda_n$ is the eigenvalue of $\rho_{\eta=0}$ corresponding to the eigenstate $|\phi_n\rangle$, and the derivative is evaluated at $\eta=0$. A large value of the QFI means a high precision in the quantum estimation of the parameter $\eta$, provided that we choose the right measurement. Notice that the optimal value of the QFI is achieved by a pure state, due to its convexity~\cite{Toth14}. A possible estimator saturating Eq.~\eqref{CR} is given by the mean of the $M$ single-copy measurement outcomes of the observable $\hat O=\frac{\hat L}{H}$, where $\hat L=2\sum_{mn}\frac{\langle\phi_m| (\partial_\eta\rho_\eta)|_{\eta=0}|\phi_n\rangle}{\lambda_m+\lambda_n}|\phi_m\rangle\langle \phi_n|$ is the symmetric logarithmic derivative of $\rho_\eta$ computed at $\eta=0$~\cite{Paris09}. This estimator is optimal for evaluating the reflective parameter in the assumed neighborhood of zero.

By using the fact that the derivative computed at $\eta=0$ is given by the trace of the commutator $(\partial_\eta\rho_\eta)|_{\eta=0} = \text{Tr}_S\left[s^\dag b-s b^\dag,|\psi\rangle_{SI}\langle\psi|\otimes \rho_{B}\right]$, and that $\rho_{\eta=0}=\sum_{\alpha}p_\alpha|v_\alpha\rangle\langle v_{\alpha}|\otimes \rho_B$ has a simple diagonal form, we can infer the following general formula for the QFI~\cite{supp}
\begin{equation}\label{Fisher}
H=\frac{4}{1+N_B}\sum_{\alpha\alpha'}\frac{p_{\alpha}p_{\alpha'}}{p_{\alpha'}+p_{\alpha}\frac{N_B}{N_B+1}}\left|\langle w_{\alpha'}|s|w_{\alpha}\rangle\right|^2.
\end{equation}
Equation~\eqref{Fisher} relates the QFI to the Schmidt vectors of the signal, and it allows us to upper bound the maximal achievable precision. First, by implementing the inequality $p_{\alpha'}/\left(p_{\alpha'}+p_{\alpha}\frac{N_B}{N_B+1}\right)\leq1$, and by using the relations $\sum_{\alpha}| w_\alpha\rangle\langle w_\alpha|=\mathbb{I}$ and $N_S=\sum_{\alpha}p_{\alpha}\langle w_{\alpha}| s^\dag s|w_{\alpha}\rangle$, we obtain $H\leq \frac{4N_S}{1+N_B}\equiv H^{(1)}_Q$. Notice that the completeness relation can be assumed by adding zero probability terms in the Schmidt decomposition. The bound $H_Q^{(1)}$ is saturated, for instance, by a two-mode squeezed state in the limit of zero photons (see the examples below). Instead, the inequality between arithmetic and geometric means $\frac{p_\alpha p_{\alpha'}}{p_{\alpha'}+p_{\alpha}\frac{N_B}{N_B+1}}\leq \sqrt{\frac{N_B+1}{N_B}}\frac{p_{\alpha}+p_{\alpha'}}{4}$ yields $H\leq \frac{2N_S+1}{N_B}\equiv H^{(2)}_Q$ by the same argument as above. The bound $H_Q^{(2)}$ is particularly useful for the usual QI situation, in which we have a bright environment $N_B\gg1$. In this case, $H_Q^{(2)}$ is attained by a classical coherent state with large number of signal photons, as we will see soon. In the following we will denote the bound on the QFI by $H_Q\equiv \min\{H_Q^{(1)},H_Q^{(2)}\}$.

\begin{figure}[t!]
\includegraphics[scale=0.70]{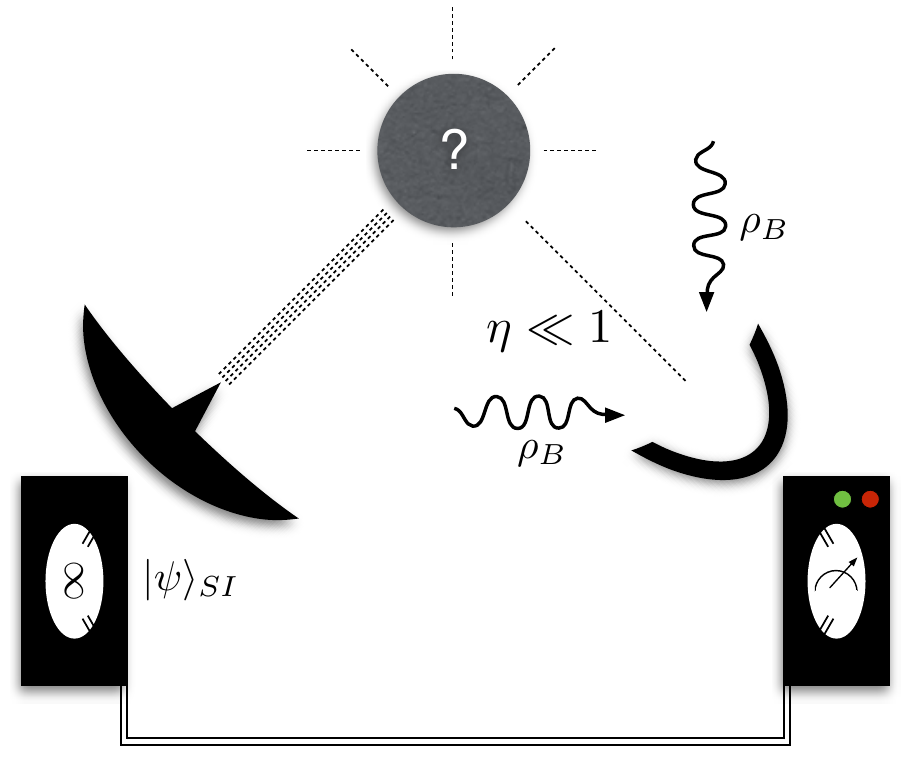}
\caption{Scheme of the quantum illumination protocol. (a) An entangled state, e.g. two-mode squeezed state, is generated in the lab. The idler beam stays in a controlled transmission line while the signal is emitted toward the object we want to detect. Since its reflectivity is small, $\eta\ll1$, most of the light captured by the receiver is thermal noise. By measuring the correlations between the signal and the idler beams, it is possible to detect the presence of an object with a smaller error probability than protocols involving classical light, with a gain up to 3~dB in the error probability exponent.}
\label{Protocol}
\end{figure}

In order to quantify the maximum possible gain achieved by an entangled state, we optimize Eq.~\eqref{Fisher} for non-entangled states. In this case, if $|\psi\rangle_S$ is the quantum state of the  signal, then $H=\frac{4\left|{}_S\langle\psi|s|\psi\rangle_{S}\right|^2}{1+2N_B}$. The last expression is maximal when $|\psi\rangle_S$ is an eigenstate of the annihilation operator, i.e. a coherent state. Therefore, we obtain $H=\frac{4N_S}{1+2N_B}\equiv H_C$, that saturates the bound $H_Q$ on the QFI for large $N_S$. Notice that, in the case when $N_B\ll1$, the QFI upper bound is saturated by a coherent state for any $N_S$, meaning that we can have a gain using nonclassical resources only when the target region is embedded in a thermal bath with a non-zero mean number of photons. From the previous inequalities, we conclude that $H/H_C\leq 2$ for low $N_S$, and $H/H_C\leq 1+\frac{1}{2N_S}$ for large $N_S$. The derived bounds imply that no structured optical device is able to go beyond the 3 dB in the reflectivity estimation problem. As shown in the examples below, this gain is saturated at least by Gaussian states and Schr\"odinger's cat states in the limit of zero signal photons. Lastly, we notice that not all entangled states are useful for estimating the reflectivity parameter. A paradigmatic case is $|\psi\rangle_{SI}=\frac{1}{\sqrt{d}}\sum_{n=0}^{d-1} |n\rangle_S| n \rangle_I$, a maximally entangled state which has the same QFI as the coherent state, as one can straightforwardly check by using Eq.~\eqref{Fisher}.

{\it Quantum Illumination.---}
In the QI protocol, we may consider a strategy based on the evaluation of the parameter $\eta$ with the estimator $\hat \eta=\frac{1}{M}\sum_{i=1}^MO_i$, where $O_i$ are the outcomes of the observable optimizing the QFI. The figure of merit used in Bayesian hypothesis testing is given by the error probability
\begin{equation}
\text{Pr}_{err}=\pi_0\Pr\left(1|H_0\right)+\pi_1\Pr\left(0|H_1\right),
\end{equation}
where $\pi_{0}$ ($\pi_1$) is the a priori probability of the absence (presence) of the object, while $\Pr\left(1|H_0\right)$ ($\Pr\left(0|H_1\right)$) is the probability to have a false positive (false negative), denoted as type I (II) error. First, we remind that $\hat \eta$ is unbiased in a neighborhood of $\eta=0$, having $\text{Tr}\,(\rho_\eta \hat O )=\eta+O(\eta^2)$. It is thus natural to define a test as follows: we declare the presence of the object whether $\hat \eta>\xi\eta$ for some $0<\xi<1$, and its absence otherwise. In this case, we have that $\Pr\left(1|H_0\right)=\Pr\left[\hat \eta >\xi\eta\right|H_0]\equiv P_{\rm I}$ and  $\Pr\left(0|H_1\right)=\Pr\left[\hat \eta -\eta<-(1-\xi)\eta\right|H_1]\equiv P_{\rm II}$. Eventually, one should choose $\xi$ in order to minimize the error probability ${\rm Pr}_{err}$. If we transmit a signal in a coherent state $|\alpha\rangle$, with $\alpha=\sqrt{N_S}e^{i\phi}$, the error probability optimized upon global measurements is given by ${\rm Pr}_{err}^C\sim e^{-\eta^2 N_S (\sqrt{N_B+1}-\sqrt{N_B})^2M}$~\cite{note}. In the $N_B\gg1$ limit, the same decaying rate is reached by the measurement that optimized the QFI, i.e. $\hat O_C=\frac{e^{-i\phi} b+e^{i\phi}b^\dag}{2\sqrt{N_S}}$, which is a quadrature operator up to a normalization factor. In fact, by using that the measured state is Gaussian for any $\eta$, we can deduce the classical type I and II error probabilities:
\begin{align}\label{classer}
P_{\rm I,\rm II}=\frac{1}{2}{\rm erfc}\left(\sqrt{\frac{\eta_{\rm I,II}^2H_CM}{2}}\right)\sim \exp\left(-\frac{ \eta_{\rm I,II}^2H_CM}{2}\right),
\end{align}
where $\eta_{\rm I}=\xi\eta$ and $\eta_{\rm II}=(1-\xi)\eta$. Both types of error decay exponentially with a rate $\eta_{\rm I,II}^2 H_CM/2$, and the optimal decaying rate of the error probability is obtained for $\xi=\frac{1}{2}$. In the $N_B\gg1$ limit, the optimal error probability approximates to ${\rm Pr}_{err}^C \approx e^{-M\eta^2 N_S/4N_B}$, whose decaying rate is the same as the one of $P_{\rm I}$ and $P_{\rm II}$ found in Eq.~\eqref{classer} if we set $\xi=\frac{1}{2}$.  In the following, we will compare a suboptimal error probability for entangled states with the optimal one for coherent states in order to show the quantum enhancement case by case. The aim of using nonclassical resources is to find a better convergence rate for the error probability ${\rm Pr}_{err}$, or to minimize $P_{\rm II}$ by keeping bounded $P_{\rm I}$ (see Refs.~\cite{Spedalieri14, Wilde16} for this analysis). In the following, we show that both types of error decay faster if an entangled state and the optimal measurement given by the QFI are used. This is proven by applying the classical Cram\'er-Chernoff theorem~\cite{Hayashi02} to the distribution of the measurement outcomes. In the following, we consider the maximally entangled states of the form $|\psi\rangle_{SI}=\sum_{n}\sqrt{p_n}|n\rangle_S|v_n\rangle_I$. 

\begin{theo*}[Type I-II error probabilities] Let $|\psi\rangle_{SI}=\sum_{n}\sqrt{p_n}|n\rangle_S|v_n\rangle_I$ be the Schmidt decomposition of the signal-idler state, and denote $\rho_S$ the state of the signal. Then $P_{\rm I,\rm II}\sim \exp\left(-\frac{\eta_{\rm I, II}^2HM}{2}\right)$ provided that $\exists C>0$ s.t. $\langle s^k s^{\dag k}\rangle_{\rho_S}\leq k!C^k$ $\forall k\in\mathbb{N}$.
\end{theo*}
\begin{proof}
In the Supplemental Material~\cite{supp} we prove that the moment generating function $M_\eta(t)=\text{Tr}\,\rho_\eta e^{t\left(\hat O-\text{Tr}\,\rho_\eta \hat O\right)}$ is finite in the interval $t\in\left[0,\sqrt{\frac{H^2N_B}{C}}\right)$, provided that $\langle s^k s^{\dag k}\rangle_{\rho_S}\leq k!C^k$. The proof is based on bounding the moments of the outcome distribution in a neighborhood of $\eta=0$. Moreover, the expression $M_\eta(t)=1+\left(\frac{1}{2H}+O(\eta)\right)t^2+O(t^3)$ holds as $t\rightarrow 0$. Now, the classical Cram\'er-Chernoff theorem says us that $-\frac{1}{M}\log P_{\rm I}\sim \sup_t\left(\eta_{\rm I} t-\log M_{\eta=0}(t)\right)$. The supremum is achieved for  $t=t^*\simeq\eta_{\rm I}H$, as $\log M_{\eta=0}(t)\simeq \frac{t^2}{2H}$ for small $t$. Therefore, we have $P_{\rm I}\sim \exp\left(-\frac{\eta_{\rm I}^2HM}{2}\right)$. Similarly, one can show that $P_{\rm II}\sim \exp\left(-\frac{\eta_{\rm II}^2HM}{2}\right)$.
\end{proof}
As a consequence, a gain in the QFI implies the same gain in the exponent of both types of error probability in the $N_B\gg1$ limit. In addition this approach holds also for discriminating between two values of the reflectivity parameter which are small, but different from zero. In this case, the optimal estimator does not change, and the $\eta$ in the error probability exponent should be replaced by the numerical difference of the two values of interest. The possible quantum advantage is kept also in this scenario. Let us note that examples of strategies based on the estimation of the parameter $\eta$ have been investigated both theoretically~\cite{Guha09} and experimentally~\cite{Lopaeva13, Zhang13,Zhang15} for the case of Gaussian states in the limit small signal photons. Here, we have extended the analysis to non-Gaussian quantum states and to any number of photons in the bath and the signal. Additionally, our analysis provides a computable upper bound for the error probability optimized upon all local quantum measurements (local strategies). Finally, a relevant advantage of our approach consists in providing explicitly the estimator which attains the aforementioned bound.

\begin{figure}[t!]
\includegraphics[scale=0.434]{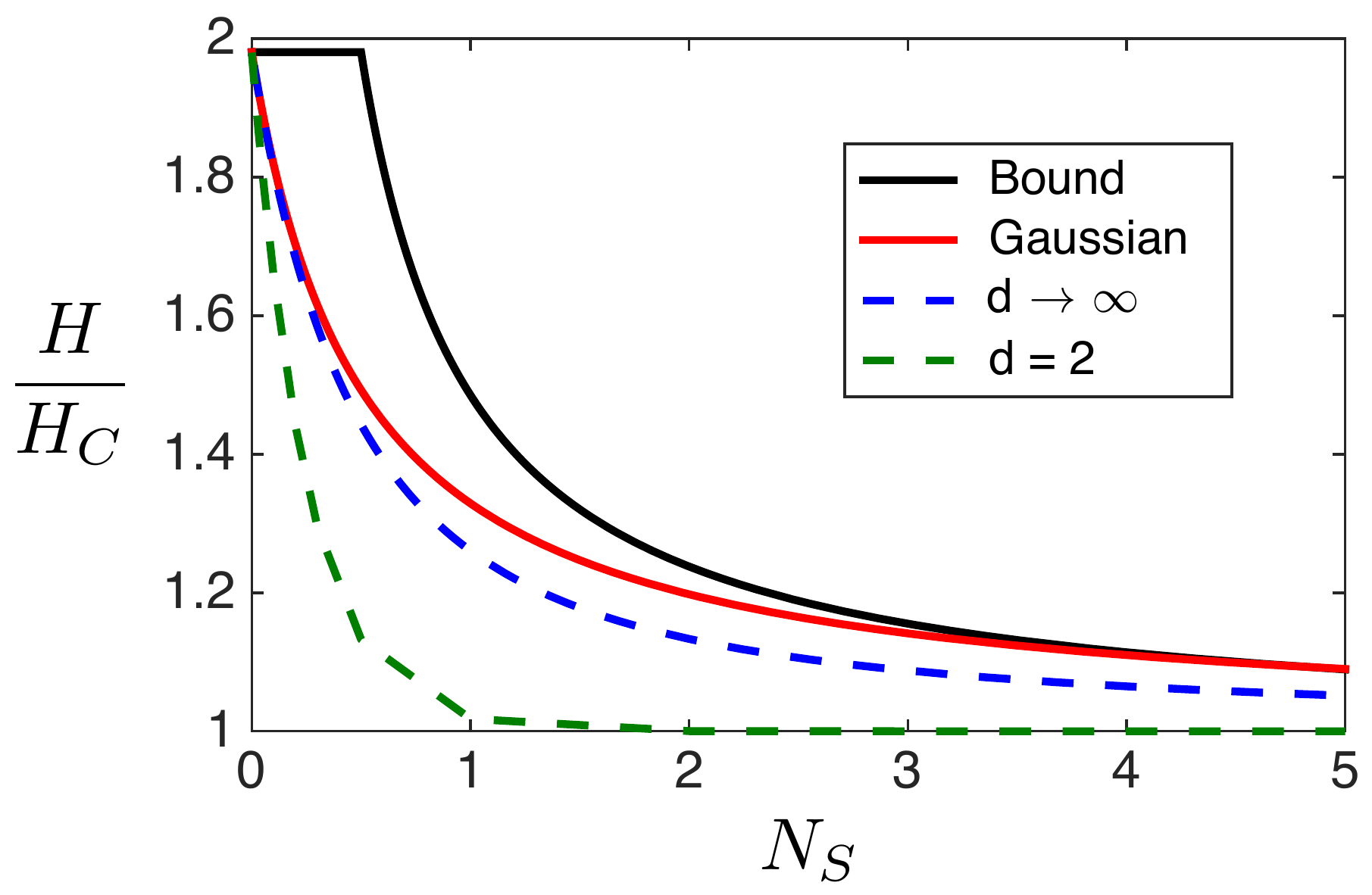}
\caption{We plot the gain in the quantum Fisher information $H$ for the Gaussian state (red line) and Schr\"odinger's cat states (dashed lines), versus the classical case $H_C$ corresponding to a coherent state transmitter. The lines corresponds to the case $N_B=50$. Both Gaussian and Schr\"odinger's cat states achieve the maximum gain in the low photon regime, but the former are sizably more stables. The calculation for the Gaussian state is exact, see Eq.~\eqref{gaussian}, while the QFI for the Schr\"odinger's cat states has been calculated numerically by truncating the Hilbert space of the received signal. This has been done by checking the convergence of the QFI for increasing Hilbert space dimension.}
\label{Curves}
\end{figure}

{\it Examples.---} Let us now illustrate our techniques by introducing a couple of paradigmatic examples achieving the maximum gain for the QI protocol in the low-photon regime, which shows that our upper bound to the optimal error probability is non-trivial. As discussed, it is sufficient to show the gain in the QFI and use the aforementioned estimator.

{\it (i) Gaussian states.---} Regarding the two-mode squeezed state case, the QFI can be analytically computed, due to the easy Schmidt decomposition, i.e. $|\psi\rangle_{SI}=\sum_{n=0}^\infty \sqrt{\frac{N_S^n}{(1+N_S)^{1+n}}}|n\rangle_S|n\rangle_I$. By using Eq.~\eqref{Fisher}, one can find that~\cite{supp}
\begin{equation}\label{gaussian}
H_{\text{Gauss}}=\frac{4N_S}{1+N_B}\frac{1}{1+\frac{N_S}{1+N_S}\frac{N_B}{1+N_B}}.
\end{equation}
Notice that $H_{\text{Gauss}}/H_C\simeq 1+\frac{1}{1+2N_S}$ for large $N_B$, so the gain achieved in this strategy decreases when $N_S$ increases. We also notice that, for $N_B\ll1$, this strategy performs worse than the optimal receiver for coherent states. This fact agrees with the previous bounds, which state that in this regime the QFI of any entangled states cannot be larger than the one corresponding to a coherent state. The optimal observable is $\hat O_{\text{Gauss}}=ab+a^\dag b^\dag$, up to a normalization factor, where $a$ is the idler mode, and $b$ is the incoming signal mode. This measurement can be realized with linear optical elements and photon counting measurements. An optical circuit implementing this measurement has been proposed in Guha et al.~\cite{Guha09} in the limit of small $N_S$.

{\it (ii) Schr\"odinger's cat state.---} In order to compare this result, we consider the Schr\"odinger's cat state $|\psi\rangle_{SI}=\frac{1}{\sqrt{d}}\sum_{k=0}^{d-1}|\alpha_k\rangle_S|w_k\rangle_I$, where $|\alpha_k\rangle$ is a coherent state with amplitude $\alpha_k=\sqrt{N_S}e^{i\frac{2\pi k}{d}}$, $\langle w_{k}| w_{k'}\rangle=\delta_{kk'}$, and $d\geq2$. These states can be generated, for instance, in a circuit QED platform. One may use the results in Ref.~\cite{Vlastakis13} to create a standard Schr\"odinger's cat ($d=2$) in the transmon-resonator system, and then use the fact that a transmon can simultaneously interact with additional resonators~\cite{Wang16}. Regarding the computation of the QFI ($H_{\text{Schr\"o}}$), we need to find the eigenvectors and the eigenvalues of the separable state $\rho_I\otimes \rho_B$. As $\rho_B$ has a simple diagonal form, let us focus on $\rho_I=\frac{1}{d}\sum_{k,k'}\langle \alpha_k| \alpha_{k'}\rangle|w_{k'}\rangle \langle w_{k}|$. We notice that $[\hat T_d, \rho_I]=0$, where $\hat T_d=\sum_{l=0}^{d-1}|w_{l}\rangle\langle w_{l+1}|$, with the convention $|w_{d}\rangle\equiv |w_{0}\rangle$, is the boost operator. It follows that $\rho_I$ has the same eigenvectors of the non-degenerate operator $\hat T_d$, i.e. $|v_k\rangle=\frac{1}{\sqrt{d}}\sum_{l=0}^{d-1}e^{i\frac{2\pi kl}{d}}|w_{l}\rangle$ ($k=0,\dots,d-1$), with corresponding eigenvalues $\lambda_k=\sqrt{d}\ \langle w_{0}| \rho_I |v_k\rangle$. It can be easily shown that $H_{\text{Schr\"o}}/H_C\simeq2$ for low $N_S$~\cite{supp}, while for finite $N_S$ a numerical calculation is needed. This can be done by truncating suitably the Hilbert space of the received signal. One should expect a higher QFI for larger dimension. For this reason, we have considered the limit of  infinite $d$, corresponding to the state $|\psi\rangle_{SI}=\sum_{n=0}^\infty\sqrt{e^{-N_S}\frac{N_S^n}{n!}}|n\rangle_S|z_n\rangle_I$ with $\langle z_{n}|z_{n'}\rangle=\delta_{nn'}$~\cite{supp}, whose form is the one considered in the Theorem. In this case, the state is already written as its Schmidt decomposition, and the QFI is obtained by plugging the probabilities obtained from Schmidt coefficients in Eq.~\eqref{Fisher}. The resulting series converges fast and can be computed efficiently up to an arbitrary small error. The results for $d=2$ and infinite $d$ are depicted in Fig.~\ref{Curves}, showing that these states perform worse than the two-mode squeezed state in the non-zero signal photon regime, but they are essentially the same in the low-photon regime. The optimal operator in this case is rather complicated, but for the case $d=2$ in the low signal photon regime, it corresponds to the measurement in the degenerate Jaynes-Cumming basis, i.e. $\hat O_{\text{Schr\"o}}=\sigma^+ b+\sigma^-b^\dag$ up to a normalization factor. This measurement can be in principle implemented by adapting the ideas of spectroscopy of the Rabi model~\cite{Ciuti, Felicetti} to the Jaynes-Cumming case.

{\it Conclusions.---} We have considered the problem of optimally estimating the reflectivity parameter of an object embedded in an environment. Our analysis shows that, using entangled states as a resource, we can obtain an advantage up to $3$~dB in the QFI with respect the optimal classical strategy. We have applied these results to the QI scenario, providing a non-trivial upper bound on the optimal error probability. This bound depends solely on the QFI of the signal-idler state, which is easily computable, and it allows us to extend the advantage of the QI protocol to a class of non-Gaussian states. Moreover, our results are not limited to a bright environment ($N_B\gg1$) and low signal photons ($N_S\ll1$) cases, but they hold for any number of photons in the bath and the signal. In the examples, we have discussed the Gaussian states and the multilevel Schr\"odinger's cat states, which also performs optimally in the low-photon regime. Indeed, recent technological advances show that Schr\"odinger's cat states can be useful for quantum computation~\cite{Ofek16}, and this makes of them a possible alternative to the Gaussian states in the QI protocol.

The authors acknowledge support from Spanish MINECO/FEDER Grant No. FIS2015-69983-P and No. FIS2015-70856-P, Basque Government Grant IT986-16, and UPV/EHU UFI 11/55 and a PhD grant,  CAM Research Network QUITEMAD+, Consejer\'ia de Educac\'ion, Juventud y Deporte, Comunidad de Madrid (S2013/ICE-2801), and the European project AQuS (Project No. 640800). The authors thank Giuseppe Vitagliano and Iagoba Apellaniz for useful discussions.

\begin{widetext}

\section{Supplemental material for  \\``Quantum Estimation Methods for Quantum Illumination''}

In this supplemental material we give the details of the claims in the main text.

\maketitle

\section{Quantum Fisher Information}

Here, we derive the formula given in the main text for the quantum Fisher information (QFI). Let us consider a general state for the signal-idler system in its Schmidt decomposition form
\begin{align}
|\psi\rangle_{SI} =  \sum_{\alpha}\sqrt{p_\alpha} | w_\alpha\rangle_S | v_\alpha\rangle_I,
\end{align}
where $\langle v_\alpha | v_{\alpha'}\rangle = \langle w_\alpha | w_{\alpha'}\rangle = \delta_{\alpha \alpha'}$. In the following, we will drop the system labels, with the convention that $|w_\alpha (v_\alpha)\rangle$ corresponds to the signal (idler) system, and $|\psi\rangle\equiv|\psi\rangle_{SI}$. We measure the state $\rho_\eta=\text{Tr}_S\,\left(U_\eta|\psi\rangle\langle\psi|\otimes \rho_{B}U_{\eta}^\dag\right)$, where $U_\eta\simeq\exp\left[\eta(s^\dag b-s b^\dag)\right]$. We have that the derivative computed at $\eta=0$ is 
\begin{align}
\partial_\eta\rho_\eta = \text{Tr}_S\left[s^\dag b-s b^\dag,|\psi\rangle\langle\psi|\otimes \rho_{B}\right]=\sum_{\alpha \alpha'}\sqrt{p_{\alpha} p_{\alpha'}}|v_\alpha\rangle \langle v_{\alpha'}|\otimes \left[\langle w_{\alpha'}|s^\dag|w_{\alpha}\rangle b - \langle w_{\alpha'}| s | w_{\alpha}\rangle] b^\dag, \rho_{B}\right].
\end{align}
Moreover, $\rho_{\eta=0}$ has the simple diagonal form $\rho_{\eta=0}=\sum_{\alpha}p_\alpha|v_\alpha\rangle\langle v_{\alpha}|\otimes \rho_B$. The eigenvalue of $\rho_{\eta=0}$ corresponding to the eigenstate $| v_\alpha\rangle |n\rangle$ is $p_\alpha\rho_n$, where $\rho_n = \frac{N_B^n}{(1+N_B)^{1+n}}$ is the corresponding eigenvalue of the thermal state $\rho_B$. The QFI for $\rho_\eta$ computed at $\eta=0$ is thus given by
\begin{align}
H&=2\sum_{\alpha \alpha' n n'}\frac{\left|\langle v_\alpha,n|\sum_{\beta\beta'}\sqrt{p_\beta p_{\beta'}} |v_\beta\rangle \langle v_{\beta'}|\otimes \left[\langle w_{\beta'}|s^\dag|w_{\beta}\rangle b - \langle w_{\beta'}| s | w_{\beta}\rangle] b^\dag, \rho_{th}\right] |v_{\alpha'},n'\rangle\right|^2}{p_\alpha\rho_n+p_{\alpha'}\rho_{n'}} \label{3}\\
\quad&=2\sum_{\alpha \alpha' n n'} \frac{p_{\alpha}p_{\alpha'}\left|\langle w_{\alpha'}|s^\dag|w_{\alpha}\rangle(\rho_{n'}-\rho_n)\sqrt{n+1}\delta_{n',n+1}- \langle w_{\alpha'}|s|w_{\alpha}\rangle(\rho_{n'}-\rho_n)\sqrt{n'+1}\delta_{n,n'+1} \right|^2}{p_\alpha\rho_n+p_{\alpha'}\rho_{n'}} \\
\quad&=2 \sum_{\alpha \alpha' n n'}\frac{p_{\alpha}p_{\alpha'}(\rho_{n'}-\rho_n)^2}{p_\alpha\rho_n+p_{\alpha'}\rho_{n'}}\left(\left|\langle w_{\alpha'}|s^\dag|w_{\alpha}\rangle\right|^2(n+1)\delta_{n',n+1} + \left|\langle w_{\alpha'}|s|w_{\alpha}\rangle\right|^2 (n'+1)\delta_{n,n'+1}\right) \\
\quad&=2\sum_{\alpha \alpha' n} (n+1)p_{\alpha}p_{\alpha'}(\rho_{n+1}-\rho_n)^2 \left(\frac{\left|\langle w_{\alpha'}|s^\dag|w_{\alpha}\rangle\right|^2}{p_\alpha\rho_n+p_{\alpha'}\rho_{n+1}}+\frac{\left|\langle w_{\alpha'}|s|w_{\alpha}\rangle\right|^2}{p_\alpha\rho_{n+1}+p_{\alpha'}\rho_{n}} \right) \\
\quad& = 4 \sum_{\alpha \alpha' n} (n+1)\rho_np_{\alpha}p_{\alpha'}\frac{\left|\langle w_{\alpha'}|s|w_{\alpha}\rangle\right|^2}{p_{\alpha'}+p_{\alpha}\frac{\rho_{n+1}}{\rho_n}}\left(1-\frac{\rho_{n+1}}{\rho_n}\right)^2 \label{7}\\
\quad&= \frac{4}{1+N_B}\sum_{\alpha\alpha'}\frac{p_{\alpha}p_{\alpha'}}{p_{\alpha'}+p_{\alpha}\frac{N_B}{N_B+1}}\left|\langle w_{\alpha'}|s|w_{\alpha}\rangle\right|^2. \label{8}
\end{align}
The equalities from \eqref{3} to \eqref{7} consists in simple algebra and rearrangement of indexes, while from \eqref{7} to \eqref{8} we have used that $\rho_{n+1}/\rho_n=N_B/(1+N_B)$. 

\section{Finiteness of the Moment generating function}
Here, we derive a bound on the moment generating function (MGF) of the distribution of the outcomes for the optimal measurement. The optimal observable is given by $\hat O=\frac{\hat L}{H}$, where $\hat L$ is the symmetric logarithmic derivative of $\rho_\eta$ computed at $\eta=0$~\cite{Paris09}. We have thus
\begin{align}
\hat O&= \frac{2}{H} \sum_{\alpha \alpha' n n'}\frac{\sqrt{p_\alpha p_{\alpha'}}(\rho_{n'}-\rho_n)\left(\langle w_{\alpha'}|s^\dag|w_{\alpha}\rangle\sqrt{n+1}\delta_{n',n+1}- \langle w_{\alpha'}|s|w_{\alpha}\rangle\sqrt{n'+1}\delta_{n,n'+1}\right)}{p_\alpha\rho_n+p_{\alpha'}\rho_{n'}}|v_\alpha, n\rangle\langle v_{\alpha'}, n'| \nonumber\\
\quad&=-\frac{2}{H(1+N_B)}\sum_{\alpha\alpha'}|v_\alpha\rangle\langle v_{\alpha'}|\otimes \left(c_{\alpha\alpha'}^*b+c_{\alpha'\alpha}b^\dag\right),
\end{align}
where we have introduced the quantity $c_{\alpha\alpha'}\equiv \frac{\sqrt{p_\alpha p_{\alpha'}}\langle w_{\alpha}|s|w_{\alpha'}\rangle}{p_{\alpha} +p_{\alpha'}\frac{N_B}{1+N_B}}$. In the following, we shall consider the case  $|w_{\alpha}\rangle\equiv |\alpha\rangle$ for all $\alpha$, i.e. the Schmidt vectors of the signal are Fock states.

The MGF $M_\eta(t)=\text{Tr}\,\rho_\eta e^{t\left(\hat O-\text{Tr}\,\rho_\eta \hat O\right)}$ is defined by the exponential power series $M_\eta (t)=\sum_{k=0}^\infty \frac{t^k\text{Tr}\,\left(\rho_\eta( \hat O - \text{Tr}\, \rho_\eta \hat O)^k\right)}{k!}$. We will prove the convergence for $M_{\eta=0}(t)$ and for $\partial_\eta M_{\eta}(t)|_{\eta=0}$ in order to ensure the convergence of the MGF in a neighborhood of $\eta=0$. 

\subsection{Convergence of $M_{\eta=0}(t)$}
We check the convergence of the exponential series defining $M_{\eta=0}(t)$ by bounding the moments $F_k\equiv\text{Tr}\,\rho_{\eta=0}\hat O^k$. Notice that $F_{2k+1}=0$ straightforwardly, as the non-central moments of the thermal state $\rho_B$ are zero. For the even moments we have
\begin{align}
\left(\frac{H(1+N_B)}{2}\right)^{2k}F_{2k}&= \sum_{\beta \alpha_1\dots \alpha_{2k}\alpha'_1\dots\alpha'_{2k}}p_{\beta}\langle v_{\beta}| v_{\alpha_1}\rangle\langle v_{\alpha'_1}| v_{\alpha_2}\rangle\cdots\langle v_{\alpha'_{2k-1}}| v_{\alpha_{2k}}\rangle \langle v_{\alpha'_{2k}}| v_{\beta}\rangle \nonumber\\
\quad&\quad\quad\quad\quad\quad\quad\times \text{Tr}\,\left(\rho_{B}\left[c_{\alpha_1\alpha'_1}^*b+c_{\alpha'_1\alpha_1}b^\dag\right]\cdots\left[c_{\alpha_{2k}\alpha'_{2k}}^*b+c_{\alpha'_{2k}\alpha_{2k}}b^{\dag}\right]\right) \\
\quad&= \sum_{\alpha_1\dots \alpha_{2k}}p_{\alpha_1}\text{Tr}\left(\rho_{B}\prod_{l=1}^{2k} \left[c_{\alpha_{l}\alpha_{l+1}}^*b+c_{\alpha_{l+1}\alpha_{l}}b^\dag\right]\right) \label{11}
\end{align}
with the notation $\alpha_{2k+1}\equiv\alpha_1$. We have then 
\begin{align}
F_{2k}&\leq \left(\frac{1}{H\sqrt{N_B(1+N_B)}}\right)^{2k}\sum_{\alpha_1\dots \alpha_{2k}}p_{\alpha_1}\text{Tr}\left(\rho_{B}\prod_{l=1}^{2k} \left[\langle \alpha_{l+1}|s^\dag|\alpha_l\rangle b+\langle \alpha_{l+1}|s|\alpha_l\rangle b^\dag\right]\right) \label{12}\\
\quad&=\left(\frac{1}{H\sqrt{N_B(1+N_B)}}\right)^{2k}\sum_{\sigma}\langle\sigma(s,s^\dag)\sigma(b,b^\dag)^\dag\rangle_{\rho_S\otimes \rho_{B}} \label{13}\\
\quad&\leq \left(\frac{1}{H\sqrt{N_B(1+N_B)}}\right)^{2k}{2k \choose k}\langle s^{k} s^{\dag k}\rangle_{\rho_{S}} \langle b^{k} b^{\dag k}\rangle_{\rho_{B}} \label{14}\\
\quad&=\left(\frac{1}{H^2N_B}\right)^{k}\frac{(2k)!}{k!}\langle s^{k} s^{\dag k}\rangle_{\rho_{S}}. \label{15}
\end{align}
Here, from \eqref{11} to \eqref{12} we have used the inequality of arithmetic and geometric means $\frac{\sqrt{p_{\alpha}p_{\alpha'}}}{p_{\alpha}+p_{\alpha'}\frac{N_B}{1+N_B}}\leq\frac{1}{2}\sqrt{\frac{1+N_B}{N_B}}$. From \eqref{12} to \eqref{13} we have used the completeness relation, the fact that non central moments of the thermal state $\rho_B$ are zero, and we have written the result in terms of $\sigma(a,a^\dag)$, denoting a possible way of arranging the product of $2k$ elements taken in the set $\{a,a^\dag\}$.  From \eqref{13} to \eqref{14}  we have used that the expected values of the antinormal ordering of the ${2k \choose k}$ central moments contributing to the sum are larger than using any other ordering, e.g $\langle a a^{\dag}a a^\dag\rangle \leq \langle a^2 a^{\dag 2}\rangle$. From \eqref{14} to \eqref{15}, we have used that $\langle b^{k} b^{\dag k}\rangle_{\rho_{B}}=k!(1+N_B)^k$. 

Next, if $\langle s^k s^{\dag k}\rangle\leq k!C^k$ for some $C>0$, then $M_{\eta=0}(t)$ is finite for $t\in\left[0,\sqrt{\frac{H^2N_B}{C}}\right)$. In fact, $|F_{2k}|\leq \left(\frac{C}{H^2N_B}\right)^k(2k)!$ implies that the series $M_{\eta=0}(t)=\sum_{k=0}^\infty \frac{t^{2k}F_{2k}}{(2k)!}$ is bounded by the geometric series $\sum_{k=0}^{\infty}\left(\frac{Ct^2}{H^2N_B}\right)^k$, that converges for $\frac{Ct^2}{H^2N_B}<1$.

\subsection{Convergence of $\partial_\eta M_{\eta}(t)|_{\eta=0}$}
In this subsection we check the convergence of the series defining the derivative of the MGF:

\begin{equation}
\partial_\eta M_{\eta}(t)|_{\eta=0}=\text{Tr}\,\partial_\eta\rho_\eta|_{\eta=0} e^{t\left(\hat O-\text{Tr}\,\rho_{\eta=0} \hat O\right)}-\text{Tr}\,\rho_{\eta=0} e^{t\left(\hat O-\text{Tr}\,\partial_\eta \rho_\eta|_{\eta=0} \hat O\right)}.
\end{equation}
The convergence for the second part of the former expression follows from the result of the previous subsection. Regarding the first part, we have that $\text{Tr}\, \rho_{\eta=0}\hat O=0$ and 
\begin{align}
\partial_\eta \rho_{\eta}|_{\eta=0}=\sum_{\beta}\sqrt{p_{\beta-1}p_\beta \beta}\left(|v_\beta\rangle\langle v_{\beta-1}|\otimes [b^\dag,\rho_B]-|v_{\beta-1}\rangle\langle v_\beta|\otimes [b,\rho_B]\right)\equiv \rho_1+\rho_2.
\end{align}
In the following, we bound the moments $G_k=\text{Tr}\,\partial_\eta \rho_{\eta=0}|_{\eta=0}\hat O^k=\text{Tr}\, \rho_1\hat O^k+\text{Tr}\,\rho_2\hat O^k\equiv G^1_k+G^2_k$. Notice that $G_{2k}=0$ as the non-central moments of the thermal state $\rho_B$ are zero. Regarding the odd moments, we have
\begin{align}
\left(\frac{H(1+N_B)}{2}\right)^{2k+1}G^1_{2k+1}&= \sum_{\beta \alpha_1\dots \alpha_{2k+1}\alpha'_1\dots\alpha'_{2k+1}}\sqrt{p_{\beta-1}p_{\beta}\beta}\langle v_{\beta-1}| v_{\alpha_1}\rangle\langle v_{\alpha'_1}| v_{\alpha_2}\rangle\cdots\langle v_{\alpha'_{2k}}| v_{\alpha_{2k+1}}\rangle \langle v_{\alpha'_{2k+1}}| v_{\beta}\rangle\nonumber\\
\quad&\quad\quad\quad\quad\quad\times \text{Tr}\left([b^\dag , \rho_{B}]\left[c_{\alpha_1\alpha'_1}^*b+c_{\alpha'_1\alpha_1}b^\dag\right]\cdots\left[c_{\alpha_{2k+1}\alpha'_{2k+1}}^*b+c_{\alpha'_{2k+1}\alpha_{2k+1}}b^{\dag}\right]\right) \\
\quad&= \sum_{\alpha_1\dots \alpha_{2k+1}}\sqrt{p_{\alpha_1}p_{\alpha_1+1}}\,\text{Tr}\left([b^\dag,\rho_{B}]\prod_{l=1}^{2k+1} \left[c_{\alpha_{l}\alpha_{l+1}}^*b+c_{\alpha_{l+1}\alpha_{l}}b^\dag\right]\right), \label{18}
\end{align}
with the notations $c_{\alpha_{2k+1} \alpha_{2k+2}}\equiv \frac{\sqrt{p_{\alpha_{1}+1} p_{\alpha_{2k+1}}}\langle {\alpha_{2k+1}}|ss^\dag| \alpha_{1}\rangle}{p_{\alpha_{2k+1}} +p_{\alpha_1+1}\frac{N_B}{1+N_B}}$ and $c_{\alpha_{2k+2} \alpha_{2k+1}}\equiv \frac{\sqrt{p_{\alpha_{1}+1} p_{\alpha_{2k+1}}}\langle {\alpha_{1}}|s^2| \alpha_{2k+1}\rangle}{p_{\alpha_{1}+1} +p_{\alpha_{2k+1}}\frac{N_B}{1+N_B}}$. Notice that $G^{1}_k$ is real, even though $\rho_1$ is not Hermitian. We have then 
\begin{align}
G^1_{2k+1}&\leq \left(\frac{2}{H(1+N_B)}\right)^{2k+1}\sum_{\alpha_1\dots \alpha_{2k+1}}\sqrt{p_{\alpha_1}p_{\alpha_1+1}}\,\text{Tr}\left(\{b^\dag,\rho_{B}\}\prod_{l=1}^{2k+1} \left[c_{\alpha_{l}\alpha_{l+1}}^*b+c_{\alpha_{l+1}\alpha_{l}}b^\dag\right]\right) \label{19} \\
\quad&\leq \frac{2}{HN_B}\left(\frac{1}{H\sqrt{N_B(1+N_B)}}\right)^{2k} {2k+1 \choose k} \langle s^{k+1}s^{\dag k+1}\rangle_{\rho_S}\langle b^{k+1}b^{\dag k+1}\rangle_{\rho_B} \label{20}\\
\quad& \leq \frac{2(1+N_B)}{HN_B}\left(\frac{1}{H^2N_B}\right)^k\frac{(2k+1)!}{k!}\langle s^{k+1}s^{\dag k+1}\rangle_{\rho_S}. \label{21}
\end{align}
Here, from \eqref{18} to \eqref{19} we have used that the two terms of the commutator contributes positively to $G^1_{2k+1}$. From \eqref{19} to \eqref{20} we have used the inequality of the arithmetic and geometric means  $\frac{\sqrt{p_{\alpha}p_{\alpha'}}}{p_{\alpha}+p_{\alpha'}\frac{N_B}{1+N_B}}\leq\frac{1}{2}\sqrt{\frac{1+N_B}{N_B}}$, the inequality $\left\{\frac{p_{\alpha_1+1}}{p_{\alpha_1+1}+p_{\alpha_{2k+1}}\frac{N_B}{1+N_B}}, \frac{p_{\alpha_1+1}}{p_{\alpha_{2k+1}}+p_{\alpha_1+1}\frac{N_B}{1+N_B}}\right\}\leq \frac{1+N_B}{N_B}$, the completeness relation, and then we have consider the antinormal ordering of the  $2k+1 \choose k $ terms contributing to the sum, in the same way as in the previous subsection. From \eqref{20} to \eqref{21} we have used that $\langle b^{k+1} b^{\dag k+1}\rangle_{\rho_{B}}=(k+1)!(1+N_B)^{k+1}$.

Notice that the same bound holds for $G^2_{2k+1}$, by using the same inequalities.
Next, if $\langle s^{k+1} s^{\dag k+1}\rangle\leq (k+1)!C^{k+1}$ for some $C>0$, then $\partial_\eta M_{\eta=0}(t)|_{\eta=0}$ is finite for $t\in\left[0,\sqrt{\frac{H^2N_B}{C}}\right)$. In fact, $|G_{2k+1}|\leq\frac{4C(1+N_B)}{HN_B}\left(\frac{C}{H^2N_B}\right)^k(2k+1)!(k+1)$ implies that the series $\partial_\eta M_{\eta=0}(t)|_{\eta=0}=\sum_{k=0}^\infty \frac{t^{2k+1}G_{2k+1}}{(2k+1)!}$ is bounded by the series $\frac{4Ct(1+N_B)}{HN_B}\sum_{k=0}^{\infty}(k+1)\left(\frac{Ct^2}{H^2N_B}\right)^k$, that converges for $\frac{Ct^2}{H^2N_B}<1$. 

This is sufficient to prove the convergence in a neighborhood of $\eta=0$, however with the same effort one could prove that the series converges for any $\eta$, by bounding further derivatives of $M_\eta(t)$ with respect $\eta$.
Finally, we have that
\begin{equation}
M_{\eta}(t)= 1+\left(\frac{1}{2H}+O(\eta)\right)t^2+O(t^3)\quad {\rm as}\quad t\rightarrow 0.
\end{equation}
The correction $O(\eta)$ to the coefficient of $t^2$ follows from the Taylor expansion of $M_\eta(t)$ in the neighborhood of $\eta=0$.

\section{Examples}

In this section, we discuss the examples given in the main text, deriving the QFI for the Guassian and the Schr\"odinger's cat states.
\subsection{Gaussian states}

Here we compute the QFI $H_{\rm{Gauss}}$ for the entangled Gaussian states $|\psi\rangle_{SI}=\sum_{n=0}^\infty \sqrt{\frac{N_S^n}{(1+N_S)^{n+1}}}|n\rangle_S | n\rangle_I$. By simply plugging the Schmidt vectors and the corresponding probabilities in QFI formula derived in Eq.~\eqref{8}, using that $|\langle n'| s |n\rangle|^2=n\delta_{n',n-1}$, and recasting, we obtain
\begin{align}
H_{\rm{Gauss}}&=\frac{4}{1+N_B}\frac{1}{1+\frac{N_S}{1+N_S}\frac{N_B}{1+N_B}}\frac{1}{1+N_S}\sum_{n}n\left(\frac{N_S}{1+N_S}\right)^n \nonumber\\
\quad&=\frac{4N_S}{1+N_B}\frac{1}{1+\frac{N_S}{1+N_S}\frac{N_B}{1+N_B}}.
\end{align}

\subsection{Schr\"odinger's cat states}

Here we derive the formula that we have used to compute numerically the QFI for the multilevel Schr\"odinger's cat state $|\psi\rangle_{SI}=\frac{1}{\sqrt{d}}\sum_{k=0}^{d-1}|\alpha_k\rangle_S|w_{k}\rangle_I$. As seen in the main text, the eigenvectors of $\rho_I$ are $|v_k\rangle=\frac{1}{\sqrt{d}}\sum_{l=0}^{d-1}e^{i\frac{2\pi kl}{d}}|w_{l}\rangle$ ($k=0,\dots,d-1$) with corresponding eigenvalues $\lambda_k=\sqrt{d}\ \langle 0| \rho_I |v_k\rangle$. This gives us the spectral decomposition of $\rho_I\otimes \rho_B$, which is the state in the case $\eta=0$. We use the definition of the QFI to achieve the formula
\begin{equation}\label{scro}
H_{\rm{Schro}}=\frac{2N_S}{d^4}\sum_{ll'nn'}\frac{\left|\sum_{r,s}\langle \alpha_r|\alpha_s\rangle e^{i\frac{2\pi (l's-lr)}{d}}(\rho_{n'}-\rho_n)\left[e^{-i\frac{2\pi s}{d}}\sqrt{n'}\delta_{n,n'-1}-e^{i\frac{2\pi r}{d}}\sqrt{n'+1}\delta_{n,n'+1}\right]\right|^2}{\rho_n\lambda_l+\rho_{n'}\lambda_{l'}}.
\end{equation}
Notice that $\langle \alpha_r|\alpha_s\rangle=\exp\left[-N_S\left(1-e^{i2\pi(s-r)/d}\right)\right]$. For low number of photons, i.e. $N_S\ll1$, $\langle \alpha_r|\alpha_s\rangle\simeq 1$ holds, and the expression in Eq.~\eqref{scro} can be readily computed:  $H_S\simeq\frac{4N_S}{1+N_B}$. For finite $N_S$ we have done a numerical computation, by summed on one of the delta indexes in order to reduce the computation times, and by truncating the received signal Hilbert space. Notice that the expression in Eq.~\eqref{scro} is monotonically increasing with this truncation dimension. We have checked, by computing Eq.~\eqref{scro} for increasing truncation dimensions, that the numerically value converges.

In the limit of infinite $d$, we obtain the state $|\psi\rangle_{SI}=\sum_{n=0}^\infty\sqrt{e^{-N_S}\frac{N_S^n}{n!}}|n\rangle_S|z_n\rangle_I$ (for some $|z_n\rangle$ with $\langle z_n|z_{n'}\rangle=\delta_{nn'}$), which is of the form considered in the hypothesis testing case. In fact, in this limit $\rho_S=\frac{1}{2\pi}\int_0^{2\pi}d\phi |\sqrt{N_S}e^{i\phi}\rangle\langle \sqrt{N_S}e^{i\phi}|=e^{-N_S}\sum_{n=0}^\infty \frac{N_S^n}{n!}|n\rangle\langle n|$ has non-degenerate eigenvalues, and $|\psi\rangle_{SI}$ has infinite Schmidt-rank.

\end{widetext}


\begin{thebibliography}{40}
\bibitem{note} The error probability exponent in the asymmetric case, corresponding to $\pi_0\not=\pi_1$, is the same as the one in the symmetric case (see Ref.~\cite{Audenaert07}). The latter has been computed in Ref.~\cite{Tan08}.
\end{thebibliography}
\end{document}